\documentclass[a4paper,11pt]{article}
\usepackage{latexsym}
\usepackage{amsmath, amsthm, amssymb, bbm}
\usepackage[mathscr]{eucal}
\usepackage{cite}
\usepackage{hyperref}
\usepackage{slashed}

\theoremstyle{plain}
\newtheorem{proposition}{Proposition}[section]

\newtheorem{lemma}[proposition]{Lemma}

\theoremstyle{definition}

\theoremstyle{remark}
\newtheorem{remark}[proposition]{Remark}

\newcommand{\rhs}{r.h.s.\ }
\newcommand{\lhs}{l.h.s.\ }
\newcommand{\wrt}{w.r.t.\ }
\newcommand{\cf}{cf.\ }

\newcommand{\ud}{\mathrm{d}}

\newcommand{\C}{\mathbb{C}}

\newcommand{\E}{\mathfrak{E}}
\newcommand{\F}{\mathfrak{F}}
\newcommand{\A}{\mathfrak{A}}

\newcommand{\D}{\mathfrak{D}}
\newcommand{\skal}[2]{\langle #1 , #2 \rangle}

\newcommand{\id}{\mathrm{id}}

\newcommand{\eps}{\varepsilon}

\newcommand{\Di}{\slashed D}

\DeclareMathOperator{\WF}{WF}
\DeclareMathOperator{\tr}{tr}

\newcommand{\ret}{{\mathrm{r}}}
\newcommand{\adv}{{\mathrm{a}}}
\newcommand{\os}{{\mathrm{os}}}

\newcommand{\fR}{{\mathfrak{R}}}

\newcommand{\Bg}{{\mathfrak{Bg}}}
\newcommand{\Emb}{{\mathfrak{Emb}}}
\newcommand{\Test}{{\mathfrak{Test}}}

\newcommand{\vp}{{\varphi}}
\newcommand{\ph}{{\varphi}}

\begin{document}

\title{Locally covariant chiral fermions and anomalies}
\author{Jochen Zahn \\  \\ Fakult\"at f\"ur Physik, Universit\"at Wien, \\ Boltzmanngasse 5, 1090 Wien, Austria. \\ jochen.zahn@univie.ac.at}

\date{\today}

\maketitle

\begin{abstract}
We define chiral fermions in the presence of non-trivial gravitational and gauge background fields in the framework of locally covariant field theory. This allows to straightforwardly compute the chiral anomalies on non-compact Lorentzian space-times, without recourse to a weak field approximation.
\end{abstract}

\section{Introduction}

The framework of \emph{locally covariant field theory} \cite{BrunettiFredenhagenVerch, HollandsWaldWick} proved extremely successful in the context of quantum field theory on curved spacetimes, \cf \cite{HollandsWaldReview, BeniniDappiaggiHackReview} for recent reviews. The framework can be straightforwardly extended to encompass more general external fields, in particular gauge potentials \cite{LocCovDirac}. As one has the freedom to shift parts of the contribution due to the external field from the free to the interacting part of the Lagrangian, one may wonder whether the two possibilities lead to equivalent theories. This question was first raised for shifts of the contribution due to the spacetime metric \cite{HollandsWaldStress}, and the equivalence of the two approaches was termed \emph{perturbative agreement} (which can be seen as a stronger form of the Ward identities). In particular it was shown that, in renormalizable theories, the only obstruction for perturbative agreement is a nonvanishing divergence of the free stress-energy tensor. Analogously, perturbative agreement can be achieved for shifts of the contribution due to the gauge potential, unless the divergence of the free current does not vanish \cite{BackgroundIndependence}. A nonvanishing divergence of the stress-energy tensor or the current is usually called an \emph{anomaly}. 

The main examples of fields with anomalous stress-energy tensor or current are chiral fermions. In the present work, we show how chiral fermions fit into the framework of locally covariant field theory and compute the anomaly of the current and the stress-energy tensor. Of course, these anomalies are well-known, \cf \cite{AlvarezGaumeGinsparg85, Bertlmann} for overviews. However, from a conceptual point of view, the corresponding calculations are not completely satisfactory. To begin with, an anomaly is often defined as the non-invariance of the
effective action under gauge transformations of the external fields, or, equivalently, as the non-vanishing of the divergence of the current derived from it. However, the definition of the effective action requires the choice of a state. But for generic background fields, there is no preferred vacuum state. In particular, this raises the question whether the anomaly is independent of the state. If it is, then it should be possible to see it already at the algebraic level, i.e., without reference to a state.

A further drawback of the usual computations of the anomaly is that they either involve ill-defined loop integrals (in the perturbative approach), or are done on compact Riemannian spaces (as in Fujikawa's \cite{Fujikawa79} or the heat kernel method \cite{VassilevichHeatKernel}), in which case the relation to the physically relevant case of  non-compact Lorentzian spacetimes remains obscure. Strictly speaking, one can not even write down a Dirac Lagrangian for chiral fermions on Riemannian spaces.

In the context of locally covariant field theory, one works on the algebraic level, so the stress-energy tensor or the current are elements of the algebra of observables (in contrast to their expectation values, which are usually considered). Anomalies of these observables then arise because non-linear fields have to be defined by point-splitting \wrt a \emph{Hadamard parametrix} $H$, which has the same singularity structure as Hadamard two-point functions, but is defined in a locally covariant manner. This entails that the parametrix is a bi-solution to the equation of motion only up to smooth remainders. It is these smooth remainders that lead to non-vanishing divergences of the stress-energy tensor or the current. Hence, the computation of the anomalies is reduced to the computation of coinciding point limits of (covariant derivatives of) the smooth remainders. It turns out that these are given by coinciding point limits of (covariant derivatives of) so-called \emph{Hadamard coefficients}. These are related to the coefficients in a formal expansion of the heat kernel
(the coefficients $b_k(x,y)$ in the notation of \cite{VassilevichHeatKernel}), 
providing a bridge to the usual heat kernel methods, \cf also \cite{HackMoretti}.

The approach we consider here provides a \emph{local} perspective on anomalies, in contrast to the \emph{global}, or even \emph{topological} viewpoint that is often emphasized, in particular inspired by the relation to the index theorem. In our approach, one can understand an anomaly as an obstruction to finding a Hadamard parametrix $H$ such that $\tr [Q H] = 0$, where the square brackets denote the coinciding point limit and $Q$ is a bi-differential operator that vanishes on bi-solutions to the equation of motion. In particular, this point of view does no longer refer to any notions of quantum physics.

For the anomaly of the current, we obtain an expression that is in formal agreement with the result obtained via the heat kernel method. The purely gravitational anomaly in dimension $n=4k+2$ is usually computed either perturbatively in a weak field approximation \cite{AlvarezGaumeWitten}, or via the index theorem in dimension $4k+4$, by considering $M_{4k+4} = M_{4k+2} \times S_2$ \cite{AlvarezGaumeGinsparg85}. We show that one can compute the purely gravitational anomaly straightforwardly, without recourse to a higher dimensional index theorem or a weak field approximation, as a coinciding point limit of (derivatives of) a Hadamard coefficient. In two dimensions, our result agrees with those obtained by other methods. It remains to be shown that this is true for all dimensions. This requires a better understanding of the coinciding point limit of derivatives of Hadamard coefficients.

In the usual terminology, we compute the \emph{covariant} anomalies. In view of the well-known relation of covariant and consistent anomalies in the path integral formalism \cite{BardeenZumino}, one would expect a close relation of the anomalies discussed here and the consistent anomalies in the Batalin-Vilkovisky formalism in perturbative algebraic quantum field theory \cite{RejznerFredenhagenQuantization}. This is a topic for future work.

The article is structured as follows. In the next section, we review the aspects of locally covariant field theory that are relevant for our discussion, using the scalar field as the illustrative example. In Section~\ref{sec:ChiralFermions}, we discuss how chiral fermions fit into this framework. We mainly refer to the extensive literature on Dirac fermions and indicate the necessary modifications. In Section~\ref{sec:Parametrix}, we review the construction of the Hadamard parametrix and prove some results on the coinciding point limit. These are then used in Section~\ref{sec:Anomalies} to compute the divergence of the current and the stress-energy tensor.

\subsection{Notation and conventions}
\label{sec:Notation}

The following conventions and notations are adopted from \cite{BGP07}: The signature is $(-, +, \dots, +)$ and the d'Alembertian defined as $\Box = - \nabla^\mu \nabla_\mu$. Minus the squared geodesic distance of $x$ and $x'$ is denoted by $\Gamma(x, x')$.

More generally, in sections of $M \times M$, the first variable will be denoted by $x$ and the second by $x'$. Accordingly, primed derivatives act on the second variable. Indices on $\Gamma$ denote covariant derivatives. The coinciding point limit of a section on $M \times M$ is denoted by square brackets.

The field strength is defined in the mathematical convention, i.e., $F_{\mu \nu} = [\nabla_\mu, \nabla_\nu]$. The spinorial curvature is denoted by $\fR_{\mu \nu}$. 

The dimension of spacetime is denoted by $n$ and assumed to be even.

\section{Locally covariant field theory}
\label{sec:Scalar}

We begin by reviewing the framework of locally covariant field theory. For simplicity, we do this for the scalar field. The crucial requirement is that a theory is not defined on a particular background, but on all possible ones, in a coherent way. This allows to speak of \emph{the same} theory on different backgrounds, and in particular to investigate the influence of changes in the background on the quantum fields. To make this mathematically precise, one defines the set $\Bg$ of backgrounds as the set of $n$-dimensional, globally hyperbolic,\footnote{For details on the notions of global hyperbolicity, we refer to \cite[Section~1.3]{BGP07}. For our purposes, the crucial point is that on such spacetimes the Klein-Gordon operator has unique retarded and advanced propagators.} oriented and time-oriented manifolds.

The notion of compatibility will be formulated by reference to embeddings that preserve as much structure as possible. To be precise, one says that $\psi \in  \Emb(M; M')$ if $\psi$ is an isometric embedding $\psi: M \to M'$, which is a diffeomorphism on its range and preserves (time-) orientation and the causal structure, i.e., all causal curves in $M'$ connecting $\psi(x)$ and $\psi(y)$ lie in $\psi(M)$.\footnote{This ensures that the pull-back of a retarded propagator on $M'$ to $M$ coincides with the retarded propagator on $M$.}

A \emph{locally covariant field theory} is now an assignment $\Bg \ni M \mapsto \A(M)$, where $\A(M)$ is a unital $*$-algebra, interpreted as the algebra of observables measurable on $M$.
This assignment is required to be consistent in the sense that for each $\psi \in  \Emb(M; M')$ there is an injective $*$ homomorphism $\alpha_{\psi}: \A(M) \to \A(M')$ of the corresponding algebras, such that
\begin{align*}
 \alpha_{\id} & = \id, &
 \alpha_{\psi} \circ \alpha_{\psi'} = \alpha_{\psi \circ \psi'}.
\end{align*}
One also wants to know what is the same observable on different spacetimes, for example, one would like to have a consistent assignment of a stress-energy tensor to all backgrounds. This is the concept of a \emph{field}. Concretely, a field $\Phi$ is an assignment $\Bg \ni M \mapsto \Phi_M$, where $\Phi_M$ is a linear map
\[
 \Phi_M : \Test(M) \to \A(M),
\]
where $\Test(M)$ is a space of compactly supported smooth test tensors on $M$. For the case of the stress-energy tensor, one would choose the space of symmetric tensors of rank $2$. This assignment is required to be compatible with the embeddings in the following sense:
\begin{equation}
\label{eq:FieldCovariance}
 \alpha_\psi \Phi_M(t) = \Phi_{M'} (\psi_* t).
\end{equation}
Here $\psi_* t$ is the push-forward of the test tensor along the embedding $\psi:M \to M'$.

\begin{remark}
The requirement \eqref{eq:FieldCovariance} entails that a field is constructed out of the local geometric data: To evaluate $\Phi_M(t)$, one could also consider $\tilde M$, the causal completion of the support of $t$, with its canonical embedding $\psi: \tilde M \to M$, and define $\Phi_M(t) = \alpha_\psi \Phi_{\tilde M}(\psi^* t)$. Hence, $\Phi_M(t)$ can only depend on the geometric data on (the causal completion of) the support of $t$. By letting the support of $t$ become arbitrarily small, one sees that, heuristically, $\Phi_M(x)$ only depends on the geometric data at $x$.
\end{remark}

In order to construct $\A(M)$ for the real scalar field, one proceeds as follows \cite{HollandsWaldWick}: We consider $\F(M)$, the space of evaluation functionals $F: \E(M) \to \C$ on the configuration space $\E(M) = C^\infty(M, \C)$, of the form
\[
 F(\vp) = \sum_k \int_{M^k} f_k(x_1, \dots, x_k) \prod_i \vp(x_i) \ud_g x_i,
\]
where the $f_k$ are compactly supported symmetric distributions, fulfilling a certain condition on their wave front set.\footnote{For an introduction to the wave front set, we refer to \cite{BrouderDangHelein}. For the present purposes, it suffices to know that it is a subset of the cotangent bundle and a refinement of the singular support of a distribution.} A convenient notation for this functional is
\[
 F = \sum_k \int_{M^k} f_k(x_1, \dots, x_k) \prod_i \phi(x_i) \ud_g x_i,
\]
where $\phi(x)$ is the point-wise evaluation functional $\phi(x)(\vp) = \vp(x)$. For $\psi \in \Emb(M; M')$, one sets $(\alpha_\psi F)(\vp') = F(\psi^* \vp')$. On $\F(M)$, one defines the involution $F^*(\vp) = \overline{F(\bar \vp)}$ and a family of products $\star_\omega$,
\begin{multline*}
  (F \star_\omega G)(\ph_0) = \\
\sum_{k=0}^\infty \frac{\hbar^k}{k!} \int \tfrac{\delta^k}{\delta \ph(x_1) \dots \delta \ph(x_k)} F|_{\ph_0} \tfrac{\delta^k}{\delta \ph(y_1) \dots \delta \ph(y_k)} G|_{\ph_0} \prod_{j=1}^k \omega(x_j,y_j) \ud_{\bar g} x_j \ud_{\bar g} y_j.
\end{multline*}
Here $\omega$ are \emph{Hadamard two-point functions}, i.e., distributional bi-solutions of the Klein-Gordon operator $P = \Box + m^2$, such that
\begin{align}
\label{eq:HadamardCom}
 \omega(x,x') - \omega(x',x) & = i \Delta(x,x'), \\
\label{eq:HadamardCon}
 \overline{\omega(x,x')} & = \omega(x',x), \\
\label{eq:HadamardWF}
\WF(\omega) & \subset C_+,
\end{align}
where $\Delta = \Delta^\ret - \Delta^\adv$ is the difference of retarded and advanced propagator of $P$ and is called the \emph{causal propagator}. $C_\pm$ is a certain subset of $T^* M^2 \setminus \{ 0 \}$, with momenta contained in $\bar V_\pm \times \bar V_\mp$ ($V_\pm$ being the cone of positive/negative energy in $T^* M$).
The condition \eqref{eq:HadamardCom} ensures that one obtains the correct commutator, due to condition \eqref{eq:HadamardCon}, $\star_\omega$ is compatible with the involution, and \eqref{eq:HadamardWF} is a replacement for the spectrum condition. These requirements entail that $\omega$ is locally of Hadamard form \cite{Radzikowski}, i.e., for $n=4$,
\begin{equation}
\label{eq:HadamardForm}
 \omega(x,x') = \frac{1}{4 \pi^2} \lim_{\eps \to 0} \left( \frac{V_0(x,x')}{\Gamma_\eps(x,x')} + V (x,x') \log \frac{\Gamma_\eps(x,x')}{\Lambda^2} \right) + W(x,x').
\end{equation}
Here $\Gamma_\eps(x,x')$ denotes minus the squared geodesic distance endowed with some $i \eps$ prescription, \cf Section~\ref{sec:Parametrix}. $V_0$, $V$, and $W$ are smooth, and $V_0$ and $V$ are constructed locally and covariantly out of the geometric data along the unique geodesic connecting $x$ and $x'$, \cf Section~\ref{sec:Parametrix}. $\Lambda$ is a length scale needed to make sense of the logarithm.

The equation of motion is implemented by dividing out the ideal $\F_0(M)$ of functionals that vanish on all solutions to the Klein-Gordon operator, $\F^\os(M) = \F(M) / \F_0(M)$, i.e., by identifying two functionals if they coincide on all solutions.

The product $\star_\omega$ depends on $\omega$, but $(\F^\os(M), \star_\omega)$ and $(\F^\os(M), \star_{\omega'})$ are isomorphic \cite{HollandsWaldWick}, $\beta_{\omega, \omega'} (F \star_{\omega'} G) = \beta_{\omega, \omega'} F \star_\omega  \beta_{\omega, \omega'} G$, where
\begin{equation*}
 \beta_{\omega, \omega'} F = \sum_{k=0}^\infty \frac{\hbar^k}{k!} \int \tfrac{\delta^{2k}}{\delta \ph(x_1) \delta \ph(y_1) \dots \delta \ph(x_k) \delta \ph(y_k)} F \prod_{j=1}^k (\omega - \omega')(x_j,y_j) \ud_{ g} x_j \ud_{ g} y_j.
\end{equation*}
So, abstractly, the algebra is independent of the choice of $\omega$. We denote this abstract algebra by $\A(M)$. The representer of $F \in \A(M)$ in $(\F^\os(M), \star_\omega)$ is denoted by $F_\omega$. For $\psi \in \Emb(M; M')$, one defines $( \alpha_\psi F )_{\omega'} = \alpha_\psi F_{\psi^* \omega'}$.

It is straightforward to define fields taking values in $\F(M)$, for example $\phi^k(x) = \phi(x)^k$, which simply takes a test function as test tensor. However, the definition of non-linear fields taking values in $\A(M)$ is more involved. The point is that one has to ensure
\begin{equation}
\label{eq:FieldConsistency}
 \Phi_M(t)_\omega = \beta_{\omega, \omega'} \Phi_M(t)_{\omega'},
\end{equation}
but $\beta_{\omega, \omega'}$ acts non-trivially on non-linear functionals. As explained in \cite{HollandsWaldWick}, one can not single out a particular two-point function $\omega$, as this would spoil local covariance. However, one can take advantage of the fact that the Hadamard parametrix $H$, i.e., the first term on the \rhs of \eqref{eq:HadamardForm}, is constructed locally out of the geometric data, and it coincides with any Hadamard two-point function, up to a smooth remainder. Hence, given a local field $\Phi$ taking values in $\F(M)$, one can define the corresponding field in $\A(M)$ by
\begin{equation}
\label{eq:hatPhi}
\hat \Phi_M(t)_\omega = \beta_{\omega, H} \Phi_M(t),
\end{equation}
which ensures \eqref{eq:FieldConsistency}. It is clear that also \eqref{eq:FieldCovariance} is fulfilled. The application of this procedure to $\phi^k$ yields the Wick powers. On a Wick square, the above amounts to point-splitting \wrt the parametrix, i.e.,
\[
 \phi^2_M(x)_\omega = \lim_{x' \to x} \left( \phi_M(x)_\omega \star_\omega \phi_M(x')_\omega - H(x,x') \right).
\]

\begin{remark}
The Hadamard parametrix is only defined locally, i.e., in a neighborhood of the diagonal of $M \times M$. This, however, is sufficient, as in \eqref{eq:hatPhi}  $\beta_{\omega, H}$ acts on a local functional, so that only the behavior of $H$ at the diagonal is relevant.
\end{remark}

\begin{remark}
\label{rem:ParametrixAmbiguity}
The Hadamard parametrix is not unique, but one may add smooth functions on $M \times M$ that are locally constructed out of the geometric data along the geodesic connecting the two points. This corresponds to the ambiguities in the definition of Wick powers discussed in \cite{HollandsWaldWick}. In particular, this freedom is in general necessary to achieve a conserved stress-energy tensor, \cf \cite{HollandsWaldStress} for the scalar field and the discussion in Section~\ref{sec:Anomalies} for case of Dirac fermions.
\end{remark}

\section{Locally covariant chiral fermions}
\label{sec:ChiralFermions}

Dirac fermions have been extensively studied in the framework of locally covariant field theory, \cf \cite{VerchSpinStatistics, DHP09, SandersDirac, LocCovDirac}, and also Majorana fermions were treated \cite{BeniniDappiaggiHackReview}. Hence, for our discussion of chiral fermions, we will mostly highlight the changes that are necessary to implement chirality, and refer the reader interested in more details to the articles mentioned above. Technically, the main complication of chiral fermions is that the Dirac operator is not an endomorphism, as it maps, for example, left-handed to right-handed fermions.

In order to describe fermions charged under a gauge group $G$ on curved spacetimes in background gauge potentials, one has to include more data into the description of the background, i.e., a spin structure $SM$ over $M$ and a principal $G$ bundle $P$ over $M$, together with a connection.
Of course the embeddings now have to respect these additional structures, \cf \cite{LocCovDirac} for details. Given these structures, and a representation $\rho$ of $G$ on a complex vector space $V$, it is straightforward to construct the Dirac bundle as the associated bundle
\[
 DM = (SM + P) \times_{\sigma \times \rho} ( \C^{2^{n/2}} \otimes V ),
\]
where $\sigma$ is the spinor representation. The orientation provides us with a chirality operator
\[
 \chi = i^{1-\frac{n}{2}} \mathrm{vol} \ \cdot,
\]
where $\cdot$ stands for the Clifford multiplication. We can use it to define the projectors $\Pi_{L/R} = (\id \mp \chi)/2$ on the left/right-handed subspaces $D_{L/R} M$. We also consider the duals $D_{L/R}^* M$ of $D_{L/R}M$, and note that the Dirac conjugation maps $D_{L/R}M$ to $D_{R/L}^* M$.\footnote{In the Riemannian case, the conjugation maps $D_{L/R} M$ to $D^*_{L/R} M$, which is the origin of the problems in defining an action for chiral fermions in that case.} In particular, the bundle
\[
 D_{L/R}^\oplus M = D_{L/R} M \oplus D_{R/L}^* M
\]
is invariant under conjugation.

The smooth sections of $D_{L/R}^\circ M$, with $\circ$ either empty, $*$, or $\oplus$, will be denoted by $\E_{L/R}^\circ(M)$ and the compactly supported ones by $\D_{L/R}^\circ(M)$. The Dirac operator $\Di$ on $\E(M)$ splits as
\[
 \Di_{L/R} = \Di \circ \Pi_{L/R} : \E_{L/R}(M) \to \E_{R/L}(M).
\]
We also define
\[
 \Di_{L/R}^\oplus = \Di_{L/R} \oplus - \Di^*_{L/R} : \E_{L/R}^\oplus(M) \to \E_{R/L}^\oplus(M),
\]
where $\Di_{L/R}^* : \E_{R/L}^*(M) \to \E_{L/R}^*(M)$ is the adjoint of $\Di_{L/R}$ \wrt the pairing $\D_{R/L}^*(M) \times \D_{R/L}(M) \to \C$. We note that there is a natural pairing $\D_{R/L}^\oplus(M) \times \E_{L/R}^\oplus(M) \to \C$, defined by
\begin{equation}
\label{eq:Pairing}
 \skal{(f_R, f'_L)}{(f_L, f'_R)} = \skal{f'_L}{f_L} + \skal{f'_R}{f_R}.
\end{equation}
Hence, $\D_{R/L}^\oplus(SM,P)$ is the natural space of test tensors for linear left/right handed fields.

If we want to describe left-handed fermions, the changes \wrt the discussion of the scalar case in Section~\ref{sec:Scalar} can be summarized as follows: The local evaluation functionals $\phi(x)$ are now maps $\phi(x): \E_L^\oplus(M) \to D_L^\oplus M_x$.\footnote{The consideration of $\E_L^\oplus(M)$ instead of $\E_L(M)$ corresponds to the usual complexification, analogously to considering $C^\infty(M, \C)$ for the real scalar field. One could also say that we are considering fields and antifields simultaneously.} One also has to implement anticommutativity of these functionals, for which we refer to \cite{RejznerFermions}. A Hadamard two-point function $\omega$ is now a distributional section of $D^\oplus_L M \times D^\oplus_L M$, which is a bi-solution of the Dirac operator $\Di_L^\oplus$. Conditions \eqref{eq:HadamardCom} and \eqref{eq:HadamardCon} are replaced by
\begin{align*}
 \omega(u, v) + \omega(v, u) & = i S_L^\oplus(u, v), \\
 \overline{\omega(u, v)} & = \omega(v^*, u^*),
\end{align*}
where $u, v \in \D_R^\oplus(M)$ and $S_L^\oplus$ is the causal propagator for $\Di_L^\oplus$, \cf the next section. The definition of fields proceeds completely analogously to the scalar case (but note that the test tensors will be right handed spinors, \cf above). Of course, one has to use a parametrix $H_L^\oplus$ for $\Di_L^\oplus$ in the definition non-linear fields, the construction of which is discussed in the next section. For later convenience, we introduce the standard notations $\psi(x)$ and $\bar \psi(x)$ for the restriction of $\phi(x)$ to $\E_L(M)$ and $\E_R^*(M)$, respectively.

\section{The parametrix}
\label{sec:Parametrix}

Let us begin by recalling how to construct retarded and advanced propagators for the Dirac operator $\Di$. One considers
\[
 P = - \Di^2 = - \nabla^\mu \nabla_\mu - \tfrac{1}{4} F_{\mu \nu} [\gamma^\mu, \gamma^\nu] + \tfrac{1}{4} R,
\]
which is a normally hyperbolic operator. It has unique retarded/advanced propagators $\Delta^{\ret/\adv}$, which are, formally and on a causal domain, given by \cite{BGP07}
\begin{equation}
\label{eq:Delta_ret_adv}
 \Delta^{\ret/\adv}(x,x') \sim \sum_{j=0}^\infty V_j(x,x') R^{\ret/\adv}_{2 j + 2}(x,x').
\end{equation}
Here the $V_j$ are the Hadamard coefficients, i.e., smooth sections of $D M \times D M$, which are recursive solutions to the \emph{transport equation}
\begin{equation}
\label{eq:TransportEquation}
 \Gamma_\mu \nabla^\mu V_k - \left( \tfrac{1}{2} \Box \Gamma - n + 2k \right) V_k = 2 k P V_{k-1}, 
\end{equation}
with the initial condition $V_0(x, x) = \id_{D M_x}$. We refer to Section~\ref{sec:Notation} for the definition of $\Gamma$ and $\Box$. The $R^{\ret/\adv}_j$ are distributions on $M \times M$, the \emph{Riesz distributions}. Note that the series on the \rhs of \eqref{eq:Delta_ret_adv} does in general not converge. However, as indicated by the symbol $\sim$, the difference of $\Delta^{\ret/\adv}$ and the series truncated after $j = n-1+N$ is of regularity $C^N$ and vanishes as $\Gamma^N$ near the light cone \cite[Thm.~2.5.2]{BGP07}. Hence, for the consideration of coinciding point limits of a finite number of derivatives, the formal expression is sufficient.

Given $\Delta^{\ret/\adv}$, the retarded/advanced propagator $S^{\ret/\adv}$ for $\Di$ is defined by
\[
 S^{\ret/\adv} = - \Di \circ \Delta^{\ret/\adv} = - \Delta^{\ret/\adv} \circ \Di.
\]
That $S^{\ret/\adv}$ is a bi-solution, or equivalently, that the second equality holds, was demonstrated in \cite{DimockDirac}. Again, one defines $S = S^\ret - S^\adv$, and sets $S^\oplus = S \oplus - S^*$, with $S^*$ being the causal propagator for $\Di^*$ (which coincides with minus the adjoint of $S$).

As $P$ commutes with $\Pi_{L/R}$, so does $\Delta^{\ret/\adv}$. Hence, $S_{L/R}^{\ret/\adv} = S^{\ret/\adv} \circ \Pi_{R/L}$ interchanges the chirality and is the retarded/advanced propagator for $\Di_{L/R}$. Analogously to the above, one defines $S_{L/R}$ and $S^\oplus_{L/R}$. Due to the duality of $\D^\oplus_{R/L}(M)$ and $\E^\oplus_{L/R}(M)$, the latter can be seen as a distributional section of $D^\oplus_{L/R} M \times D^\oplus_{L/R} M$.

As for the retarded/advanced propagators, the Hadamard parametrix $H$ for $\Di$ will be defined via the Hadamard parametrix $h$ for $P$. Concretely, we have
\begin{equation}
\label{eq:hFormal}
 h^\pm \sim \frac{1}{2 \pi} \sum_{j=0}^\infty V_j T^\pm_{2j+2},
\end{equation}
where the distributions $T^\pm_j$ are defined as follows (for even $j$ and $n$):
\begin{equation}
\label{eq:T}
 T^\pm_j = \lim_{\eps \to +0}
\begin{cases}
 C'_{j,n} (-\Gamma \mp i \varepsilon \theta_0)^{\frac{j-n}{2}} & \text{ if } j < n, \\
 C_{j,n}  \Gamma^{\frac{j-n}{2}} \log (-\Gamma \mp i \varepsilon \theta_0) / \Lambda^2 & \text{ if } j \geq n,
\end{cases}
\end{equation}
where $\Lambda$ is again a length scale and
\begin{align*}
 C_{j,n} & = \frac{2^{1-j} \pi^{\frac{2-n}{2}}}{(\frac{j}{2}-1)! (\frac{j-n}{2})!}, &
 C'_{j, n} & = - \frac{2^{1-j} \pi^{\frac{2-n}{2}} (\frac{n-j}{2}-1)!}{(\frac{j}{2}-1)!}.
\end{align*}
We also used the notation $\theta_0(x,x') = t(x)-t(x')$, where $t$ is some time function. We note that $T^+_j(x,x') = T^-_j(x',x)$ and
\[
 T^+_j - T^-_j = 2 \pi i \left( R^\ret_j - R^\adv_j \right),
\]
which ensures \eqref{eq:HadamardCom}. Furthermore, the wave front sets of  the $T^+_j$ are such that \eqref{eq:HadamardWF} holds.\footnote{This holds even though \eqref{eq:Delta_ret_adv} and \eqref{eq:hFormal} are only formal expansions, \cf \cite{LocCovDirac}.}

To describe chiral fermions, we define $h_{L/R}^\pm = h^\pm \circ \Pi_{L/R}$, where $h^\pm$ is interpreted as an operator on sections of $DM$. The parametrix $H^\oplus_L$, i.e., the distributional section on $D^\oplus_L M \times D^\oplus_L M$ is then defined as
\begin{multline}
\label{eq:H_oplus}
 H^\oplus_L((f_R, f'_L), (g_R, g'_L)) = - \tfrac{1}{2} \left( h^+_R(\Di_R^* f'_L, g_R) + h^+_L(f'_L,\Di_R g_R) \right. \\
 \left. - h^-_R(\Di_R^* g'_L, f_R) - h^-_L(g'_L, \Di_R f_R) \right).
\end{multline}
Note that a distributional section on $D^\oplus_L M \times D^\oplus_L M$ can be naturally evaluated on test sections of $D^\oplus_R M \times D^\oplus_R M$, due to the canonical pairing.

We now discuss some properties of the Hadamard parametrices $h^\pm$ which will be important for the computation of the anomalies. We begin by stating the following lemma, whose proof is straightforward:
\begin{lemma}
\label{lemma}
The distributions $T^\pm_j$ defined in \eqref{eq:T} satisfy
\begin{align}
\label{eq:GammaT}
 \Gamma T^\pm_j & = \begin{cases} j (j - n + 2) T^\pm_{j + 2} & \text{if } j \neq n-2 \\ - C'_{n-2,n} & \text{if } j = n-2 \end{cases} \\
\label{eq:nablaT}
 2 j \nabla T^\pm_{j + 2} & =  \begin{cases} T^\pm_j \nabla \Gamma & \text{if } j < n \\  T^\pm_j \nabla \Gamma + 2 j C_{j+2,n} \Gamma^{\frac{j-n}{2}} \nabla \Gamma & \text{if } j \geq n \end{cases} \\
\label{eq:T0}
 T^\pm_0 & = 0.
\end{align}
For a smooth function $V$ on $M^2$, vanishing at coinciding points, we define
\[
  V \tilde T^\pm_0 = \lim_{\varepsilon \to +0} D_n V (-\Gamma \mp i \varepsilon \theta_0 + \varepsilon^2)^{-\frac{n}{2}}
\]
with
\[
 D_n =
\begin{cases}
(n-2) C'_{2,n} & \text{if } n \neq 2 \\
- 2 C_{2,2} & \text{if } n = 2.
\end{cases}
\]
Then
\begin{align}
\label{eq:nablaTildeT0}
 \nabla T^\pm_2 & = \tfrac{1}{2} \tilde T^\pm_0 \nabla \Gamma \\
\label{eq:GammaTildeT0}
 \Gamma \tilde T^\pm_0 & =
\begin{cases}
(2-n) T^\pm_2 & \text{ if $n > 2$} \\
2 C_{2,2} & \text{ if $n=2$.}
\end{cases}
\end{align}
\end{lemma}

\begin{remark}
Similar relations hold for the Riesz distributions $R^{\ret/\adv}$. The differences are that for the Riesz distributions the contributions involving $C$ and $C'$ in \eqref{eq:GammaT}, \eqref{eq:nablaT}, and \eqref{eq:GammaTildeT0} are absent, and that, instead of \eqref{eq:T0}, one has $R^{\ret/\adv}_0 = \delta$. The latter has the consequence that $\Delta^{\ret/\adv}$ are Green's functions (instead of solutions), whereas the former lead to $h^\pm$ being a solution only up to smooth remainders. We also note that the smooth remainders in \eqref{eq:GammaT}, \eqref{eq:nablaT}, and \eqref{eq:GammaTildeT0} are absent in odd dimensions. In particular, there are then no anomalies.
\end{remark}

The following proposition was proven in \cite{MorettiStressEnergy}. For the convenience of the reader, we include a proof here, too.
\begin{proposition}
\label{prop:CP_limit}
The parametrix $h^\pm$ defined in \eqref{eq:hFormal} fulfills
\begin{align}
\label{eq:Ph}
 2 \pi [P h^\pm] & = \left( C_n + 2 n C_{n+2,n} \right) [V_{\frac{n}{2}}], \\
\label{eq:nablaPh}
 2 \pi [\nabla_\mu P h^\pm] & = \left( C_n + 2 (n+2) C_{n+2, n} \right) [\nabla_\mu V_{\frac{n}{2}}], \\
\label{eq:nabla'Ph}
 2 \pi [\nabla'_{\mu} P h^\pm] & = \left( C_n + 2n C_{n+2, n} \right) [\nabla'_{\mu} V_{\frac{n}{2}}] - 4 C_{n+2,n} [\nabla_\mu V_{\frac{n}{2}}].
\end{align}
Here
\[
 C_n = \begin{cases} C_{2,2} & \text{if } n = 2, \\
 - \frac{1}{n(n-2)} C'_{n-2,n} & \text{if } n \geq 4. \end{cases}
\]
\end{proposition}
\begin{proof}
With the help of Lemma~\ref{lemma}, we compute
\begin{align*}
 P(V_0 T^\pm_2) & = P V_0 T^\pm_2 - \left( \nabla^\mu V_0 \Gamma_\mu + \tfrac{1}{2} g^{\mu \nu} \Gamma_{\mu \nu} V_0 + n V_0 \right) \tilde T^\pm_0, \\
 P(V_1 T^\pm_4) & = P V_1 T^\pm_4 - \tfrac{1}{2} \nabla^\mu V_1 \Gamma_\mu T^\pm_2 %\\
  %& \quad 
- \tfrac{1}{4} V_1 g^{\mu \nu} \Gamma_{\mu \nu} T^\pm_2 - \tfrac{1}{8} V_1 \Gamma_\mu \Gamma^\mu \tilde T^\pm_0 \\
  & \quad - \delta_{2,n} C_{4,n} \nabla^\mu \left( V_1 \Gamma_\mu \right), \\
 P(V_j T^\pm_{2j+2}) & = P V_j T^\pm_{2+2j} - \tfrac{1}{2j} \nabla^\mu V_j \Gamma_\mu T^\pm_{2j} \\
  & \quad - \tfrac{1}{4j} V_j \left( g^{\mu \nu} \Gamma_{\mu \nu} T^\pm_{2j} + \tfrac{1}{4(j-1)} T^\pm_{2j-2} \Gamma_\mu \Gamma^\mu \right) \\
  & \quad - \theta_{2j,n} C_{2j+2,n} \left( 2 \nabla^\mu V_j \Gamma^{j-\frac{n}{2}} \Gamma_\mu + V_j \nabla^\mu ( \Gamma^{j-\frac{n}{2}} \Gamma_\mu ) \right) .
\end{align*}
Using the identity $\Gamma_\mu \Gamma^\mu = - 4 \Gamma$ \cite[Eq.~(57)]{Poisson03}, we thus obtain
\begin{multline*}
 2 \pi P h^\pm \sim - U_0 \tilde T^\pm_0 - \sum_{j=1}^\infty \tfrac{1}{2j} U_j T^\pm_{2j} \\
  + C_n V_{\frac{n}{2}} - \sum_{j=n/2}^\infty C_{2j+2,n} \left( 2 \nabla^\mu V_j \Gamma^{j-\frac{n}{2}} \Gamma_\mu +V_j \nabla^\mu ( \Gamma^{j-\frac{n}{2}} \Gamma_\mu) \right).
\end{multline*}
with
\begin{equation*}
 U_j = \Gamma_\mu \nabla^\mu V_j + \tfrac{1}{2} V_j g^{\mu \nu} \Gamma_{\mu \nu} - (2j-n) V_j - 2 j P V_{j-1}.
\end{equation*}
Note that the $U_j$ vanish, by the transport equation \eqref{eq:TransportEquation}. Using \cite[Sect.~2.4]{Poisson03}
\begin{align*}
 [\Gamma_{\mu \nu}] & = - 2 g_{\mu \nu}, & [\Gamma_{\mu \nu'}] & = 2 g_{\mu \nu},
\end{align*}
we obtain \eqref{eq:Ph}, \eqref{eq:nablaPh}, and \eqref{eq:nabla'Ph}.
\end{proof}

A problem that one encounters when computing the divergence of currents in fermionic theories is that one not only finds expressions of the form treated in the above proposition, but also coinciding point limits of the form $[\Di {\Di'}^* h]$, where ${\Di'}^*$ is the adjoint Dirac operator acting on the second variable. This difficulty was already encountered in \cite{DHP09}, where the conformal anomaly and the divergence of the stress-energy tensor were computed for Dirac fermions in $n=4$ and a flat background connection. There, the problem was dealt with in a way that is not directly generalizable to chiral fermions, non-trivial gauge background fields, and arbitrary $n$. Our treatment below applies for any dimension and also simplifies considerably the proof of the results obtained in \cite{DHP09}.

A first thing to note is that $\Di \circ \Delta^{\ret/\adv} = \Delta^{\ret/\adv} \circ \Di$, as both sides of the equation coincide with $S^{\ret/\adv}$, which is unique. By the relation of the distributions $R^{\ret/\adv}$ and $T^{\pm}$,  we know that $\Di h^\pm$ and ${\Di'}^* h^\pm$ must coincide up to a smooth remainder (a different argument was given in \cite{DHP09}). Let us denote it by $J^\pm$, i.e.,
\[
 J^\pm = (\Di - {\Di'}^*) h^\pm.
\]
Hence, we have
\[
 [\Di {\Di'}^* h^\pm] = - [P h^\pm] - [\Di J^\pm].
\]
The first term on the \rhs is known from Prop.~\ref{prop:CP_limit}. It remains to compute the second. Thus, let us study $J^\pm$ in detail. We have
\begin{equation}
\label{eq:J}
 2 \pi J^\pm \sim Y_0 \tilde T^\pm_0 + \sum_{j=1}^\infty \tfrac{1}{2j} Y_j  T^\pm_{2j} + \sum_{j=\frac{n}{2}}^\infty C_{2j+2,n} \Gamma^{j-\frac{n}{2}} \Gamma_\mu [\gamma^\mu, V_j],
\end{equation}
where
\begin{align}
\label{eq:Y_j}
 Y_j & =  \tfrac{1}{2} \Gamma_\mu [\gamma^\mu, V_j] + 2 j (\Di - {\Di'}^*) V_{j-1}.
\end{align}
Here we used the notation
\[
 [\gamma_\mu, V_j](x,x') = \gamma_\mu V_j(x,x') - V_j(x,x') \gamma^{\mu'} g^\mu_{\mu'}(x,x'),
\]
where $g(x,x')$ denotes the parallel transport of tangent vector along the unique geodesic. Note that we used $\nabla'_{\mu} \Gamma = - g_{\mu'}^\mu \nabla_\mu \Gamma$, \cf \cite[Sect.~2.3.2]{Poisson03}.

The last term on the \rhs of \eqref{eq:J} is smooth, whereas the first two terms are a priori singular at $\Gamma = 0$. From the above argument, we know that their sum must be smooth, but this is in general no great help, due to the smooth remainder in \eqref{eq:GammaT}: In order to compute $[J^\pm]$, one would have to determine the coinciding point limit of up to $n-2j$ derivatives of $Y_j$. However, it turns out that the $Y_j$ all vanish, leaving us with only the third term on the \rhs of \eqref{eq:J}. First of all, $Y_0 = 0$, as $V_0$ is a scalar multiple of the parallel transport. Then, due to $[\nabla_\mu V_0] = 0$, we also have $[Y_1] = 0$. The statement then follows from the following:
\begin{lemma}
For $j \geq 1$, the $Y_j$ defined in \eqref{eq:Y_j} fulfill the transport equation
\begin{equation}
\label{eq:TransportEquationY}
 \Gamma_\mu \nabla^\mu Y_j - \left(\tfrac{1}{2} \Box \Gamma - n + 2(j-1) \right) Y_j - 2 j P Y_{j-1} = 0.
\end{equation}
\end{lemma}

\begin{proof}
Denote the \lhs of the equation by $E$ and compute
\begin{align*}
 E & = \tfrac{1}{2} \Gamma^\mu \Gamma_{\mu \nu} [\gamma^\nu, V_j] + \tfrac{1}{2} \Gamma_\lambda \Gamma_\mu [\gamma^\lambda, \nabla^\mu V_j] + 2 j \Gamma_\mu \nabla^\mu (\Di - {\Di'}^*) V_{j-1} \\
 & \quad - \left(\tfrac{1}{2} \Box \Gamma - n + 2(j-1) \right) Y_j - j P (\Gamma_\mu [\gamma^\mu, V_{j-1}]) \\
 & \quad - 4 j (j-1) (\Di - {\Di'}^*) P V_{j-2},
\end{align*}
where we used $g^{\lambda}_{\ \lambda' ; \mu} \Gamma^\mu = 0$, \cf \cite[Sect.~2.3.2]{Poisson03}, for the first term. Using $\Gamma^\lambda \Gamma_\lambda = - 4 \Gamma$ on the first term,
\begin{equation}
\label{eq:Comm_P_gamma}
 [P, \gamma^\lambda] = - \tfrac{1}{4} F_{\mu \nu} [[\gamma^\mu, \gamma^\nu], \gamma^\lambda] = - 2 \gamma_\mu F^{\mu \lambda},
\end{equation}
and inserting the transport equation \eqref{eq:TransportEquation} in the second and last term, we obtain
\begin{align*}
E & = - \Gamma_{\lambda} [\gamma^\lambda, V_j] + j \Gamma_\lambda [\gamma^\lambda, P V_{j-1}] + \tfrac{1}{2} \left(\tfrac{1}{2} \Box \Gamma - n + 2j \right) \Gamma_\lambda [\gamma^\lambda, V_j] \\
& \quad + 2 j \Gamma_\mu \nabla^\mu (\Di - {\Di'}^*) V_{j-1} - \left(\tfrac{1}{2} \Box \Gamma - n + 2(j-1) \right) Y_j \\
& \quad - j \Box \Gamma_\mu \gamma^\mu V_{j-1} - j \Box \Gamma_{\mu'} V_{j-1} \gamma^{\mu'} + 2 j \Gamma_{\mu \nu} \gamma^\nu \nabla^\mu V_{j-1} + 2 j \Gamma_{\mu \nu'} \nabla^\mu V_{j-1} \gamma^{\nu'} \\
& \quad - j \Gamma_\mu [\gamma^\mu, P V_{j-1}] + 2 j \Gamma_\mu \gamma_\nu F^{\nu \mu} V_{j-1} \\
& \quad - 2 j (\Di - {\Di'}^*) \left( \Gamma_\mu \nabla^\mu V_{j-1} - \left( \tfrac{1}{2} \Box \Gamma - n + 2(j-1) \right) V_{j-1} \right).
\end{align*}
Writing out $Y_j$, and commuting various differential operators, this simplifies to
\begin{align*}
E & = - 2 j \left(\tfrac{1}{2} \Box \Gamma - n + 2(j-1) \right) (\Di - {\Di'}^*) V_{j-1} \\
& \quad - 2 j \Gamma_{\mu \nu} \gamma^\nu \nabla^\mu V_{j-1} - 2 j \Gamma_{\mu \nu'}  \nabla^\mu V_{j-1} \gamma^{\nu'} - j \Box \Gamma_\mu \gamma^\mu V_{j-1} - j \Box \Gamma_{\mu'} V_{j-1} \gamma^{\mu'} \\
& \quad + 2 j \Gamma_{\mu \nu} \gamma^\nu \nabla^\mu V_{j-1} + 2 j \Gamma_{\mu \nu'} \nabla^\mu V_{j-1} \gamma^{\nu'} + 2 j \Gamma_\mu \gamma_\nu F^{\nu \mu} V_{j-1} \\
& \quad - 2 j \Gamma_\mu \gamma_\nu (F^{\nu \mu} + \fR^{\nu \mu}) V_{j-1} + 2 j (\Di - {\Di'}^*) \left( \left( \tfrac{1}{2} \Box \Gamma - n + 2(j-1) \right) V_{j-1} \right) \\
& = - j \Box \Gamma_\mu \gamma^\mu V_{j-1} + j \nabla_\mu \Box \Gamma \gamma^\mu V_{j-1} - 2 j \Gamma_\mu \gamma_\nu \fR^{\nu \mu} V_{j-1} \\
& = j R_{\mu \nu} \Gamma^\nu \gamma^\mu V_{j-1} - j \Gamma_\mu \gamma_\nu R^{\mu \nu} V_{j-1} \\
& = 0.
\end{align*}
Here we used the identity $\gamma^\mu \fR_{\mu \nu} = \frac{1}{2} \gamma^\mu R_{\mu \nu}$ for the spin curvature.
\end{proof}

We summarize our result in the following:
\begin{proposition}
\label{prop:D-D'_h}
Let $h^\pm$ be the parametrix \eqref{eq:hFormal} for $P = - \Di^2$. Then
\begin{equation}
\label{eq:D-D'_h}
 \Di h^\pm - {\Di'}^* h^\pm \sim \tfrac{1}{2 \pi}  \sum_{j=\frac{n}{2}}^\infty C_{2j+2,n} \Gamma^{j-\frac{n}{2}} \Gamma_\mu [\gamma^\mu, V_j].
\end{equation}
\end{proposition}

\begin{remark}
This can be straightforwardly generalized to the case of massive fermions, i.e., for $\Di = \gamma^\mu \nabla_\mu + m$. However, one should then define $P = - \Di \tilde \Di$, with $\tilde \Di = \gamma^\mu \nabla_\mu - m$. Then we still have \eqref{eq:Comm_P_gamma} and, due to $\Di \tilde \Di = \tilde \Di \Di$, also $P \Di = \Di P$ and $P \tilde \Di = \tilde \Di P$. It follows that \eqref{eq:TransportEquationY} still holds. Furthermore, in the definition of $Y_j$, one may of course replace $\Di$ and ${\Di'}^*$ by their tilded counterparts, and analogously in \eqref{eq:D-D'_h}. The result \eqref{eq:D-D'_h} then simplifies and generalizes considerably the results of \cite[Prop.~A.1]{DHP09}, as there only the case $n=4$ with flat background connection was treated, and some identities were only derived for traces.
\end{remark}

A straightforward consequence of Proposition~\ref{prop:D-D'_h} is
\[
 2 \pi [\Di (\Di - {\Di'}^*) h^\pm] = - 2 C_{2n+2,n} \gamma_\mu [\gamma^\mu, V_{\frac{n}{2}}],
\]
and similarly for supplementary derivatives. Denoting by $\tr_D$ the partial trace over the spinor indices, we thus obtain:
\begin{proposition}
Let $h^\pm$ be the parametrix \eqref{eq:hFormal} for $P = - \Di^2$. Then
\begin{align}
\label{eq:DhD}
 2 \pi \tr_D [\Di {\Di'}^* h^\pm] \chi & = \left( - C_n + 2 n C_{n+2,n} \right) \tr_D [V_{\frac{n}{2}}] \chi, \\
 2 \pi \tr_D [\nabla_\mu \Di {\Di'}^* h^\pm] \chi & = \left( - C_n + 2 n C_{n+2, n} \right) \tr_D [\nabla_\mu V_{\frac{n}{2}}] \chi, \\
 2 \pi \tr_D [\nabla'_\mu \Di {\Di'}^* h^\pm] \chi & = \left( - C_n + 2 n C_{n+2, n} \right) \tr_D [\nabla'_\mu V_{\frac{n}{2}}] \chi.
\end{align}
\end{proposition}

We finish this section by noting that the Hadamard coefficients fulfill $(V_k)^* = V_k^*$ \cite[Thm.~6.4.1]{Friedlander}, where $V_k^*$ is the Hadamard coefficient for $P^*$. Furthermore, $[h^\pm \circ P] = [ P^* {h^*}^\pm]^*$, where $h^*$ is the parametrix for $P^*$. From this and Prop.~\ref{prop:CP_limit}, we obtain
\begin{proposition}
\label{prop:CP_limit'}
The parametrix $h^\pm$ defined in \eqref{eq:hFormal} fulfills
\begin{align*}
 2 \pi [{P'}^* h^\pm] & = \left( C_n + 2 n C_{n+2,n} \right) [V_{\frac{n}{2}}], \\
 2 \pi [\nabla'_\mu {P'}^* h^\pm] & = \left( C_n + 2 (n+2) C_{n+2, n} \right) [\nabla'_\mu V_{\frac{n}{2}}], \\
 2 \pi [\nabla_{\mu} {P'}^* h^\pm] & = \left( C_n + 2n C_{n+2, n} \right) [\nabla_{\mu} V_{\frac{n}{2}}] - 4 C_{n+2,n} [\nabla'_\mu V_{\frac{n}{2}}].
\end{align*}
\end{proposition}

Finally, we note Synge's rule \cite[Sect.~2.2]{Poisson03}, i.e.,
\begin{equation}
\label{eq:SyngesRule}
 \nabla_\mu [V] = [\nabla_\mu V] + [\nabla'_\mu V].
\end{equation}

\section{Anomalies}
\label{sec:Anomalies}

We now have at our disposal all the results that are needed to compute the chiral anomalies. Let us begin with the anomaly of the current. The divergence of the current is, in the notation introduced in Section~\ref{sec:ChiralFermions}, given by
\[
 \nabla_\mu j_I^\mu = \nabla_\mu \left( \bar \psi T_I \gamma^\mu \psi \right) = - (\Di^* \bar \psi) T_I \psi + \bar \psi T_I \Di \psi.  
\]
Here $I$ is a Lie algebra index and $T_I$ the corresponding generator.
It is clear that this vanishes classically. However, the corresponding quantum field, defined analogously to \eqref{eq:hatPhi}, need not vanish, as the parametrix is in general only a bi-solution modulo smooth sections. Using the form \eqref{eq:H_oplus} of the parametrix, we obtain, for left-handed fermions,
\[
 \nabla_\mu \hat{\jmath}_I^\mu = \tfrac{\hbar}{2} \tr \left( T_I \left( [\Di_R {\Di'}_L^* h^-_R] + [ {\Di'}_L^* {\Di'}_R^* h^-_L] - [\Di_L \Di_R h^-_R] - [\Di_L {\Di'}_R^* h^-_L] \right) \right).
\]
With our definition of $h^\pm_{L/R}$, we thus obtain
\[
 \nabla_\mu \hat{\jmath}_I^\mu = \tfrac{\hbar}{2} \tr \left( T_I \left( [P h^-] - [\Di {\Di'}^* h^-] \right) \chi \right),
\]
where we used that $[{P'}^* h^\pm] = [P h^\pm]$, \cf above. With \eqref{eq:Ph} and \eqref{eq:DhD}, we obtain
\[
 \nabla_\mu \hat{\jmath}_I^\mu = \tfrac{\hbar}{2 \pi} C_n \tr \left( T_I [V_{\frac{n}{2}}] \chi \right).
\]
For right-handed fermions, the sign is reversed. Noting that, up to normalization, the integral over the trace of the $[V_k]$ corresponds to the heat kernel coefficients, this is in agreement with the expression of the anomaly in the heat kernel framework, \cf \cite{VassilevichHeatKernel}. Concretely, we have\footnote{For the case of a flat gauge connection, these were computed in \cite{Christensen78}. The modifications due to a non-trivial gauge connection are straightforward.}
\begin{align*}
 [V_1] & = - \tfrac{1}{12} R + \tfrac{1}{4} F_{\mu \nu} [\gamma^\mu, \gamma^\nu], \\
 [V_2] & = \tfrac{1}{16} F_{\mu \nu} F_{\lambda \rho} [\gamma^\mu, \gamma^\nu] [\gamma^\lambda, \gamma^\rho] - \tfrac{1}{24} R F_{\mu \nu} [\gamma^\mu, \gamma^\nu] - \tfrac{1}{12} \Box F_{\mu \nu} [\gamma^\mu, \gamma^\nu] \\ 
 & + \tfrac{1}{144} R^2 + \tfrac{1}{60} \Box R - \tfrac{1}{90} R_{\mu \nu} R^{\mu \nu} + \tfrac{1}{90} R_{\mu \nu \lambda \rho} R^{\mu \nu \lambda \rho} \\ 
 &  + \tfrac{1}{6} (\fR_{\mu \nu} + F_{\mu \nu}) ( \fR^{\mu \nu} + F^{\mu \nu})
\end{align*}
so that, with $\fR_{\mu \nu} = R_{\mu \nu \lambda \rho} \gamma^\lambda \gamma^\rho$, we obtain
\[
 \nabla_\mu \hat{\jmath}^\mu_I =
\begin{cases}
\tfrac{1}{4\pi} \tfrac{1}{\sqrt{-g}} \eps^{\mu \nu} \tr_V T_I F_{\mu \nu} & \text{ for } n=2 ,\\ \tfrac{i}{32 \pi^2} \tfrac{1}{\sqrt{-g}} \eps^{\mu \nu \lambda \rho} \tr_V T_I \left( F_{\mu \nu} F_{\lambda \rho} + \tfrac{1}{24} R_{\sigma \xi \mu \nu} R^{\sigma \xi}_{\ \ \ \lambda \rho} \right) & \text{ for } n=4.
\end{cases}
\]
In these equations, the \rhs is not of the form $\nabla_\mu Q^\mu$ for some vector field $Q$ defined locally and covariantly, so no redefinition of the parametrix can eliminate these, \cf Remark~\ref{rem:ParametrixAmbiguity}. Hence, these constitute an anomaly.

Let us now compute the purely gravitational anomaly \cite{AlvarezGaumeWitten}. We restrict to a flat background gauge connection and compute the divergence of the stress-energy tensor \cite{ForgerRomer}
\begin{equation*}
 T_{\mu \nu} = \tfrac{1}{2} \left( \bar \psi \gamma_{(\mu} \nabla_{\nu)} \psi - \nabla_{(\mu} \bar \psi \gamma_{\nu)} \psi - g_{\mu \nu} \left( \bar \psi \Di \psi + \Di^* \bar \psi \psi \right) \right).
\end{equation*}
For its divergence, one obtains
\begin{align*}
 \nabla^\mu T_{\mu \nu} & = \tfrac{1}{4} \left( - \Di^* \bar \psi \nabla_\nu \psi + \nabla_\nu \Di^* \bar \psi \psi - \nabla_\nu \bar \psi \Di \psi + \bar \psi \nabla_\nu \Di \psi \right. \\ 
 & \qquad \left. - \Di^* \Di^* \bar \psi \gamma_\nu \psi + \bar \psi \gamma_\nu \Di \Di \psi \right).
\end{align*}
Here we always have at least one Dirac operator acting on a $\psi$, so that the expression vanishes classically. For its quantum counterpart, one obtains, using the same method as above and \eqref{eq:SyngesRule},
\begin{multline}
\label{eq:divT}
 \nabla^\mu \hat T_{\mu \nu} = \tfrac{1}{8 \pi} \left( ( C_n + 2 C_{n+2,n}) \tr \left( [\nabla_\nu V_{\frac{n}{2}}] \chi - [\nabla'_\nu V_{\frac{n}{2}}] \chi \right) \right. \\
 \left. + C_{n+2,n} \tr \nabla^\mu [V_\frac{n}{2}] \chi [\gamma_\nu, \gamma_\mu] + 4 C_{n+2,n} \tr \nabla_\nu [V_{\frac{n}{2}}] \right).
\end{multline}
The last term on the \rhs is of the form $\nabla^\mu Q_{\mu \nu}$, with $Q$ a covariant symmetric tensor. Such a term can be eliminated by a redefinition of the parametrix, \cf \cite{LocCovDirac}.\footnote{Note that this is not possible for scalar fields in $n=2$, \cf \cite{HollandsWaldStress}.} As the remaining terms involve the chirality $\chi$, this shows that for Dirac fermions, the parametrix may be defined such that the stress-energy tensor is conserved, in any dimension. The second term on the \rhs of \eqref{eq:divT} can also be written in the form $\nabla^\mu Q_{\mu \nu}$, but with an anti-symmetric $Q$. It can thus not be absorbed in a redefinition of the parametrix, and constitutes a contribution to the anomaly. Also the first term contributes to the anomaly. Let us check that for $n=2$, one recovers the usual chiral gravitational anomaly: Using
\begin{align*}
 [\nabla_\mu V_1] & = \tfrac{1}{2} \nabla_\mu [V_1] + \tfrac{1}{6} \nabla^\nu \left( \fR_{\mu \nu} + F_{\mu \nu} \right), \\
 [\nabla'_\mu V_1] & = \tfrac{1}{2} \nabla_\mu [V_1] - \tfrac{1}{6} \nabla^\nu \left( \fR_{\mu \nu} + F_{\mu \nu} \right),
\end{align*}
one finds
\[
 \nabla^\mu \hat T_{\mu \nu} = \tfrac{\hbar}{96 \pi} \tfrac{r}{\sqrt{-g}} \eps_{\nu \mu} \nabla^\mu R
\]
for the first two terms on the \rhs of \eqref{eq:divT}.
Here $r$ is the dimension of the representation $\rho$. Up to an imaginary factor, this coincides with the well-known result for the purely gravitational anomaly \cite[Eq.~(12.606)]{Bertlmann}.\footnote{The expression seems to deviate from the result in \cite{AlvarezGaumeWitten} by a factor $\tfrac{1}{2}$. This, however, is due to different normalization of the stress-energy tensor, as explained in \cite[Footnote~7]{AlvarezGaumeGinsparg85}.}

Of course also the conformal anomaly can be computed in the framework employed here, \cf \cite{DHP09} for the case $n=4$.

\subsection*{Acknowledgments}

I would like to thank Christoph Stephan for helpful discussions. This work was supported by the Austrian Science Fund (FWF) under the contract P24713. Also the kind hospitality of the Erwin Schr\"odinger Institute during the workshop ``Algebraic Quantum Field Theory: Its Status and Its Future'' is gratefully acknowledged.

%\bibliography{../mybib}{}
%\bibliographystyle{../h-elsevier_new}

\end{document}